\documentclass{article}

\usepackage{amsthm,amsmath,amssymb}
\usepackage{url}
\usepackage{enumerate}
\usepackage{subcaption}
\usepackage[utf8]{inputenc}
\usepackage{xspace}
\usepackage{xcolor}
\usepackage{graphicx}
\usepackage[colorlinks=true,citecolor=black,linkcolor=black,urlcolor=blue]{hyperref}
\usepackage{xargs} 
\usepackage[colorinlistoftodos,prependcaption,textsize=tiny]{todonotes}

\theoremstyle{plain}
\newtheorem{theorem}{Theorem}

\newtheorem{corollary}[theorem]{Corollary}

\theoremstyle{definition}

\def\abs#1{\left| #1 \right|}
\def\paren#1{\left( #1 \right)}
\def\acc#1{\left\{ #1 \right\}}

\def\ceil#1{\left\lceil #1 \right\rceil}

\renewcommand{\le}{\leqslant}
\renewcommand{\ge}{\geqslant}

\title{On some interesting ternary formulas}

\author{Pascal Ochem\footnote{LIRMM, CNRS, Universit\'e de Montpellier, France. ochem@lirmm.fr} 
\and Matthieu Rosenfeld\footnote{LIP, ENS de Lyon, CNRS, UCBL, Universit\'e de Lyon, France. matthieu.rosenfeld@ens-lyon.fr}}
\begin{document}

\maketitle
\setcounter{footnote}{0}
\begin{abstract}
We obtain the following results about the avoidance of ternary formulas.
Up to renaming of the letters, the only infinite ternary words
avoiding the formula $ABCAB.ABCBA.ACB.BAC$ (resp.\ $ABCA.BCAB.BCB.CBA$) have the same set of recurrent factors as
the fixed point of $\texttt{0}\mapsto\texttt{012}$, $\texttt{1}\mapsto\texttt{02}$, $\texttt{2}\mapsto\texttt{1}$.
The formula $ABAC.BACA.ABCA$ is avoided by polynomially many binary words and there exists arbitrarily many
infinite binary words with different sets of recurrent factors that avoid it.
If every variable of a ternary formula appears at least twice in the same fragment, then the formula is $3$-avoidable.
The pattern $ABACADABCA$ is unavoidable for the class of $C_4$-minor-free graphs with maximum degree~$3$.
This disproves a conjecture of Grytczuk.
The formula $ABCA.ACBA$, or equivalently the palindromic pattern $ABCADACBA$, has avoidability index $4$.
\end{abstract}

\textbf{Acknowledgements}: This work was partially supported by the ANR project CoCoGro (ANR-16-CE40-0005).

\section{Introduction}\label{sec:intro}

A \emph{pattern} $p$ is a non-empty finite word over an alphabet
$\Delta=\acc{A,B,C,\dots}$ of capital letters called \emph{variables}.
An \emph{occurrence} of $p$ in a word $w$ is a non-erasing morphism $h:\Delta^*\to\Sigma^*$
such that $h(p)$ is a factor of $w$.
The \emph{avoidability index} $\lambda(p)$ of a pattern $p$ is the size of the
smallest alphabet $\Sigma$ such that there exists an infinite word
over $\Sigma$ containing no occurrence of $p$.

A variable that appears only once in a pattern is said to be \emph{isolated}.
Following Cassaigne~\cite{Cassaigne1994}, we associate a pattern $p$ with the \emph{formula} $f$
obtained by replacing every isolated variable in $p$ by a dot.
For example, the pattern $AABCABBDBBAA$ gives the formula $AAB.ABB.BBAA$.
The factors that are separated by dots are called \emph{fragments}.
So $AAB$, $ABB$, and $BBAA$ are the fragments of $AAB.ABB.BBAA$.

An \emph{occurrence} of a formula $f$ in a word $w$ is a non-erasing morphism $h:\Delta^*\to\Sigma^*$
such that the $h$-image of every fragment of $f$ is a factor of $w$.
As for patterns, the avoidability index $\lambda(f)$ of a formula $f$ is the size of the
smallest alphabet allowing the existence of an infinite word containing no occurrence of $f$.
Clearly, if a formula $f$ is associated with a pattern $p$,
every word avoiding $f$ also avoids $p$, so $\lambda(p)\le\lambda(f)$.
Recall that an infinite word is \emph{recurrent} if every finite factor appears
infinitely many times.
If there exists an infinite word over $\Sigma$ avoiding $p$,
then there exists an infinite recurrent word over $\Sigma$ avoiding $p$.
This recurrent word also avoids $f$, so that $\lambda(p)=\lambda(f)$.
Without loss of generality, a formula is such that no variable is isolated
and no fragment is a factor of another fragment.
We say that a formula $f$ is \emph{divisible} by a formula $f'$ if $f$ does not avoid $f'$,
that is, there is a non-erasing morphism $h$ such that the image of any fragment of $f'$ under $h$ is a factor of a fragment of $f$.
If $f$ is divisible by $f'$, then every word avoiding $f'$ also avoids $f$.
Let $\Sigma_k=\acc{0,1,\ldots,k-1}$ denote the $k$-letter alphabet. We denote by $\Sigma_k^n$ the $k^n$ words of length $n$ over $\Sigma_k$.

A formula is \emph{binary} if it has at most 2 variables.
We have recently determined the avoidability index of every binary formula~\cite{OchemRosenfeld2016}.
This exhaustive study led to the discovery of some interesting binary formulas that are
avoided by only a few binary words.
Determining the avoidability index of every ternary formula would be huge task.
However, we have identified some interesting ternary formulas and this paper describes their properties.

We say that two infinite words are equivalent if they have the same set of factors.
Let $b_3$ be the fixed point of $\texttt{0}\mapsto\texttt{012}$, $\texttt{1}\mapsto\texttt{02}$, $\texttt{2}\mapsto\texttt{1}$.
A famous result of Thue~\cite{Ber94,Thue06,Thue:1912} can be stated as follows:

\begin{theorem}~\cite{Ber94,Thue06,Thue:1912}
\label{thm:thue}
Every recurrent ternary word avoiding $AA$, \texttt{010}, and \texttt{212} is equivalent to $b_3$.
\end{theorem}

In Section~\ref{sec:b3}, we obtain a similar result for $b_3$ by forbidding one ternary formula but without forbidding explicit factors in $\Sigma_3^*$.
In Section~\ref{sec:2}, we describe the set of binary words avoiding $ABACA.ABCA$ and $ABAC.BACA.ABCA$.
We show that these formulas are avoided by polynomially many binary words and that there exist infinitely many
recurrent binary words with different sets of recurrent factors that avoid them.
In the terminology of~\cite{OchemRosenfeld2016}, these formulas are not essentially avoided by a finite set of words.
In Section~\ref{sec:nice}, we consider \emph{nice} formulas.
A formula $f$ is nice if for every variable $X$ of $f$, there exists a fragment of $f$
that contains $X$ at least twice. This notion generalizes to formulas the notion of a \emph{doubled} pattern
(that is, a pattern that contains every variable at least twice).
Every doubled pattern is $3$-avoidable~\cite{O16}.
We show that every ternary nice formula is $3$-avoidable.
In Section~\ref{sec:disprove}, we show that $ABACADABCA$ is a $2$-avoidable pattern that is unavoidable on graphs with maximum degree $3$.
In Section~\ref{sec:pal}, we show that there exists a palindromic pattern with index $4$.

A preliminary version of this paper, without the results in Sections~\ref{sec:nice} and~\ref{sec:pal},
has been presented at WORDS 2017.

\section{Formulas closely related to $b_3$}\label{sec:b3}

For every letter $c\in\Sigma_3$, $\sigma_c:\Sigma_3^*\mapsto\Sigma_3^*$ is the morphism such that
$\sigma_c(a)=b$, $\sigma_c(b)=a$, and $\sigma_c(c)=c$ with $\acc{a,b,c}=\Sigma_3$. 
So $\sigma_c$ is the morphism that fixes $c$ and exchanges the two other letters.

We consider the following formulas.
\begin{itemize}
 \item $f_b=ABCAB.ABCBA.ACB.BAC$
 \item $f_1=ABCA.BCAB.BCB.CBA$
 \item $f_2=ABCAB.BCB.AC$
 \item $f_3=ABCA.BCAB.ACB.BCB$
 \item $f_4=ABCA.BCAB.BCB.AC.BA$
\end{itemize}

Notice that $f_b$ is divisible by $f_1$, $f_2$, $f_3$, $f_4$.

\begin{theorem}
\label{mainth}
Let $f\in\acc{f_b,f_1,f_2,f_3,f_4}$.
Every ternary recurrent word avoiding $f$ is equivalent to $b_3$, $\sigma_0(b_3)$, or $\sigma_2(b_3)$.
\end{theorem}

By considering divisibility, we can deduce that Theorem~\ref{mainth} holds for $72$ ternary formulas.
Since $b_3$, $\sigma_0(b_3)$, and $\sigma_2(b_3)$ are equivalent to their reverse,
Theorem~\ref{mainth} also holds for the $72$ reverse ternary formulas.

\begin{proof}
Using Cassaigne's algorithm~\cite{cassaignealgo}, we have checked that $b_3$ avoids $f_i$, for $1\le i\le4$.
By symmetry, $\sigma_0(b_3)$ and $\sigma_2(b_3)$ also avoid $f_i$.

Let $w$ be a ternary recurrent word $w$ avoiding $f_b$.
Suppose to get a contradiction that $w$ contains a square $uu$. 
Then there exists a non-empty word $v$ such that $uuvuu$ is a factor of $w$.
Thus, $w$ contains an occurrence of $f_b$ given by the morphism $A\mapsto u,B\mapsto u,C\mapsto v$.
This contradiction shows that $w$ is square-free.

An occurrence $h$ of a ternary formula over $\Sigma_3$ is said to be \emph{basic} if $\acc{h(A),h(B),h(C)}=\Sigma_3$.
As already noticed by Thue~\cite{Ber94}, no infinite ternary word avoids squares and \texttt{012}.
So, every infinite ternary square-free word contains the $6$ factors obtained by letter permutation of \texttt{012}.
Thus, an infinite ternary square-free word contains a basic occurrence of $f_b$ if and only if
it contains the same basic occurrence of $ABCAB.ABCBA$.
Therefore, $w$ contains no basic occurrence of $ABCAB.ABCBA$.

A computer check shows that the longest ternary words avoiding $f_b$, squares, \texttt{021020120}, \texttt{102101201}, and \texttt{210212012} have length $159$.
So we assume without loss of generality that $w$ contains \texttt{021020120}.

Suppose to get a contradiction that $w$ contains \texttt{010}.
Since $w$ is square-free, $w$ contains \texttt{20102}.
Moreover, $w$ contains the factor of \texttt{20120} of \texttt{021020120}.
So $w$ contains the basic occurrence $A\mapsto\texttt{2}$, $B\mapsto\texttt{0}$, $C\mapsto\texttt{1}$ of $ABCAB.ABCBA$.
This contradiction shows that $w$ avoids \texttt{010}.

Suppose to get a contradiction that $w$ contains \texttt{212}.
Since $w$ is square-free, $w$ contains \texttt{02120}.
Moreover, $w$ contains the factor of \texttt{021020} of \texttt{021020120}.
So $w$ contains the basic occurrence $A\mapsto\texttt{0}$, $B\mapsto\texttt{2}$, $C\mapsto\texttt{1}$ of $ABCAB.ABCBA$.
This contradiction shows that $w$ avoids \texttt{212}.

Since $w$ avoids squares, \texttt{010}, and \texttt{212}, Theorem~\ref{thm:thue} implies that $w$ is equivalent to $b_3$.
By symmetry, every ternary recurrent word avoiding $f_b$ is equivalent to $b_3$, $\sigma_0(b_3)$, or $\sigma_2(b_3)$.
\end{proof}

\section{Avoidability of $ABACA.ABCA$ and $ABAC.BACA.ABCA$}\label{sec:2}
Following the terminology in~\cite{OchemRosenfeld2016}, we say that a finite set of infinite words $\cal{M}$ \emph{essentially avoids} a formula $f$
if every infinite word over $\Sigma_{\lambda(f)}$ avoiding $f$ has the same set of recurrent factors as a word in $\cal{M}$.
In this terminology, Theorem~\ref{mainth} says that $\acc{b_3,\sigma_0(b_3), \sigma_2(b_3)}$ essentially avoids many ternary formulas.
Let $b_4$ be the fixed point of $\texttt{0}\mapsto\texttt{01}$, $\texttt{1}\mapsto\texttt{21}$, $\texttt{2}\mapsto\texttt{03}$, $\texttt{3}\mapsto\texttt{23}$,
let $b_4'$ be obtained from $b_4$ by exchanging \texttt{0} and \texttt{1}, and
let $b_4''$ be obtained from $b_4$ by exchanging \texttt{0} and \texttt{3}.
Then $\acc{b_3,b_3',b_3''}$ essentially avoids $AB.AC.BA.CA.CB$~\cite{BNT89}.
Finally, five binary formulas~\cite{OchemRosenfeld2016} are known to be essentially avoided by a finite set of binary morphic words.

Thus, every formula that is known to be avoided by polynomially many words is actually essentially avoided by a finite set of morphic words.
In this section, we show in particular that this is not the case for $ABACA.ABCA$ and $ABAC.BACA.ABCA$.

We consider the morphisms $m_a:$ $\texttt{0}\mapsto\texttt{001}$, $\texttt{1}\mapsto\texttt{101}$ and
$m_b:$ $\texttt{0}\mapsto\texttt{010}$, $\texttt{1}\mapsto\texttt{110}$. That is, $m_a(x)=x\texttt{01}$ and $m_b(x)=x\texttt{10}$ for every $x\in\Sigma_2$.


We construct the set $S$ of binary words as follows:
\begin{itemize}
 \item $\texttt{0}\in S$.
 \item If $v\in S$, then $m_a(v)\in S$ and $m_b(v)\in S$.
 \item If $v\in S$ and $v'$ is a factor of $v$, then $v'\in S$.
\end{itemize}

\begin{theorem}
\label{polyth}
Let $f\in\acc{ABACA.ABCA,ABAC.BACA.ABCA}$. The set of words $u$ such that $u$ is recurrent in an infinite binary word avoiding $f$ is $S$.
\end{theorem}

\begin{proof}

Let $R$ be the set of words $u$ such that $u$ is recurrent in an infinite binary word avoiding $ABACA.ABCA$.
Let $R'$ be the set of words $u$ such that $u$ is recurrent in an infinite binary word avoiding $ABAC.BACA.ABCA$.
An occurrence of $ABACA.ABCA$ is also an occurrence of $ABAC.BACA.ABCA$, so that $R'\subseteq R$.

Let us show that $R\subseteq S$.
We study the small factors of a recurrent binary word $w$ avoiding $ABACA.ABCA$.
Notice that $w$ avoid the pattern $ABAAA$ since it contains the occurrence $A\mapsto A$, $B\mapsto B$, $C\mapsto A$ of $ABACA.ABCA$.
Since $w$ contains recurrent factors only, $w$ also avoids $AAA$.


A computer check shows that the longest binary words avoiding $ABACA.ABCA$, $AAA$, \texttt{1001101001}, and \texttt{0110010110} have length $53$.
So we assume without loss of generality that $w$ contains \texttt{1001101001}.

Suppose to get a contradiction that $w$ contains \texttt{1100}.
Since $w$ avoids $AAA$, $w$ contains \texttt{011001}.
Then $w$ contains the occurrence $A\mapsto\texttt{01},B\mapsto\texttt{1},C\mapsto\texttt{0}$ of $ABACA.ABCA$.
This contradiction shows that $w$ avoids \texttt{1100}.

Since $w$ contains \texttt{0110}, the occurrence $A\mapsto\texttt{0},B\mapsto\texttt{1},C\mapsto\texttt{1}$ of $ABACA.ABCA$ shows that $w$ avoids \texttt{01010}.
Similarly, $w$ contains \texttt{1001} and avoids \texttt{10101}.

Suppose to get a contradiction that $w$ contains \texttt{0101}.
Since $w$ avoids \texttt{01010} and \texttt{10101}, $w$ contains \texttt{001011}.
Moreover, $w$ avoids $AAA$, so $w$ contains \texttt{10010110}.
Then $w$ contains the occurrence $A\mapsto\texttt{10},B\mapsto\texttt{0},C\mapsto\texttt{1}$ of $ABACA.ABCA$.
This contradiction shows that $w$ avoids \texttt{0101}.

So $w$ avoids every factor in $\acc{\texttt{000},\texttt{111},\texttt{0101},\texttt{1100}}$.
Thus, it is to check that if we extend any factor \texttt{01} in $w$ to three letters to the right,
we get either \texttt{01001} or \texttt{01101}, that is, $\texttt{01}x\texttt{01}$ with $x\in\Sigma_2$.
This implies that $w$ is the $m_a$-image of some binary word.

Obviously, the image by a non-erasing morphism of a word containing a formula also contains the formula.
Thus, the pre-image of $w$ by $m_a$ also avoids $ABACA.ABCA$. This shows that $R\subseteq S$.

Let us show that $S\subseteq R'$, that is, every word in $S$ avoids $ABAC.BACA.ABCA$.
We suppose to get a contradiction that a finite word $w\in S$ avoids $ABAC.BACA.ABCA$ and that $m_a(w)$ contains an occurrence $h$ of $ABAC.BACA.ABCA$.

If we write $w=w_0w_1w_2w_3\ldots$, then the word
$m_a(w)=w_0\texttt{01}w_1\texttt{01}w_2\texttt{01}w_3\texttt{01}\ldots$ is such that:
\begin{itemize}
\item Every factor \texttt{00} occurs at position $0\pmod{3}$.
\item Every factor \texttt{01} occurs at position $1\pmod{3}$.
\item Every factor \texttt{11} occurs at position $2\pmod{3}$.
\item Every factor \texttt{10} occurs at position $0$ or $2\pmod{3}$, depending on whether a factor $\texttt{1}w_i\texttt{0}$ is \texttt{100} or \texttt{110}.
\end{itemize}
We say that a factor $s$ is \emph{gentle} if either $|s|\ge3$ or $s\in\acc{\texttt{00},\texttt{01},\texttt{11}}$.
By previous remarks, all the occurrences of the same gentle factor have the same position modulo 3.

First, we consider the case such that $h(A)$ is gentle.
This implies that the distance between two occurrences of $h(A)$ is $0\pmod{3}$.
Because of the repetitions $h(ABA)$, $h(ACA)$, and $h(ABCA)$ are contained in the formula, we deduce that
\begin{itemize}
\item $|h(AB)|=|h(A)|+|h(B)|\equiv0\pmod{3}$.
\item $|h(AC)|=|h(A)|+|h(C)|\equiv0\pmod{3}$.
\item $|h(ABC)|=|h(A)|+|h(B)+|h(C)|\equiv0\pmod{3}$.
\end{itemize}
This gives $|h(A)|\equiv|h(B)|\equiv|h(C)|\equiv0\pmod{3}$.
Clearly, such an occurrence of the formula in $m_a(w)$ implies an occurrence of the formula in $w$, which is a contradiction.

Now we consider the case such that $h(B)$ is gentle.
If $h(CA)$ is also gentle, then the factors $h(BACA)$ and $h(BCA)$ imply that $|h(A)|\equiv0\pmod{3}$. Thus, $h(A)$ is gentle and the first case applies.
If $h(CA)$ is not gentle, then $h(CA)=\texttt{10}$, that is, $h(C)=\texttt{1}$ and $h(A)=\texttt{0}$.
Thus, $m_a(w)$ contains both $h(BAC)=h(B)\texttt{01}$ and $h(BCA)=h(B)\texttt{10}$. 
Since $h(B)$ is gentle, this implies that \texttt{01} and \texttt{10} have the same position modulo $3$, which is impossible.

The case such that $h(C)$ is gentle is symmetrical.
If $h(AB)$ is gentle, then $h(ABAC)$ and $h(ABC)$ imply that $|h(A)|\equiv0\pmod{3}$.
If $h(AB)$ is not gentle, then $h(A)=\texttt{1}$ and $h(B)=\texttt{0}$. Thus, $m_a(w)$ contains both $h(ABC)=\texttt{01}h(C)$ and $h(BAC)=\texttt{10}h(C)$.
Since $h(C)$ is gentle, this implies that \texttt{01} and \texttt{01} have the same position modulo $3$, which is impossible.

Finally, if $h(A)$, $h(B)$, and $h(C)$ are not gentle, then the length of the three fragments of the formula is $2|h(A)|+|h(B)|+|h(C)|\le 8$.
So it suffices to consider the factors of length at most $8$ in $S$ to check that no such occurrence exists.

This shows that $S\subseteq R'$. Since $R'\subseteq R\subseteq S\subseteq R'$, we obtain $R'=R=S$, which proves Theorem~\ref{polyth}. 
\end{proof}

\begin{corollary}
Neither $ABACA.ABCA$ nor $ABAC.BACA.ABCA$ is essentially avoided by a finite set of morphic words.
\end{corollary}
\begin{proof}
Let $c(n)=\abs{S\cup\Sigma_2^n}$ denote the number of words of length $n$ in $S$.
By construction of~$S$,
$$c(n)=2\sum_{0\le i\le2}c\paren{\ceil{\tfrac{n-i}3}}\text{ for every }n\ge8.$$

Thus $c(n)=\Theta\paren{n^{\ln 6/\ln 3}}=\Theta\paren{n^{1+\ln 2/\ln 3}}$.
Devyatov~\cite{Devyatov} has recently shown that the factor complexity (i.e. the number of factors of length $n$) of a morphic word
is either $O\paren{n\ln(n)}$ or $\Theta\paren{n^{1+1/k}}$ for some integer $k\ge1$.
Thus, $S$ cannot be the union of the factors of a finite number of morphic words.
\end{proof}

\section{Ternary nice formulas}\label{sec:nice}

Clark~\cite{Clark} introduced the notion of \emph{$n$-avoidance basis} for formulas,
which is the smallest set of formulas with the following property:
for every $i\le n$, every avoidable formula with $i$ variables is divisible by at least one formula
with at most $i$ variables in the $n$-avoidance basis.
See~\cite{Clark,circular} for more discussions about the $n$-avoidance basis.
The avoidability index of every formula in the $3$-avoidance basis has been determined:
\begin{itemize}
 \item $AA$ ($\lambda=3$~\cite{Thue06})
 \item $ABA.BAB$ ($\lambda=3$~\cite{Cassaigne1994})
 \item $ABCA.BCAB.CABC$ ($\lambda=3$~\cite{circular})
 \item $ABCBA.CBABC$ ($\lambda=2$~\cite{circular})
 \item $ABCA.CABC.BCB$ ($\lambda=3$~\cite{circular})
 \item $ABCA.BCAB.CBC$ ($\lambda=3$, reverse of $ABCA.CABC.BCB$)
 \item $AB.AC.BA.CA.CB$ ($\lambda=4$~\cite{BNT89})
\end{itemize}


Recall that a formula $f$ is \emph{nice} if for every variable $X$ of $f$,
there exists a fragment of $f$ that contains $X$ at least twice.
Every formula in the $3$-avoidance basis except $AB.AC.BA.CA.CB$
is both nice and $3$-avoidable. This raised the question in~\cite{circular} whether every
nice formula is $3$-avoidable, which would generalize the $3$-avoidability of doubled patterns.
In this section, we answer this question positively for ternary formulas.

\begin{theorem}
\label{nice}
Every nice formula with at most $3$ variables is $3$-avoidable.
\end{theorem}

We say that a nice formula is minimal if it is not divisible by another nice formula with at most the same number of variables.
The following property of every minimal nice formula is easy to derive.
If a variable $V$ appears as a prefix of a fragment $\phi$, then
\begin{itemize}
 \item $V$ is also a suffix of $\phi$,
 \item $\phi$ contains exactly two occurrences of $V$,
 \item $V$ is neither a prefix nor a suffix of any fragment other than $\phi$,
 \item Every fragment other than $\phi$ contains at most one occurrence of $V$.
\end{itemize}

Thus, if $f$ is a minimal nice formula with $n\ge2$ variables, then $f$ has at most $n$ fragments.
Moreover, every fragment has length at most $2+2^{n-1}-1=2^{n-1}+1$, since otherwise it would contain
a doubled pattern as a factor.

This implies an algorithm to list the minimal nice formulas with at most $n$ variables.
The table below lists the formulas that need to be shown $3$-avoidable, that is, 
the minimal nice formulas with at most $3$ variables that do not belong to the $3$-avoidance basis.
Also, if two distinct formulas are the reverse of each other, then only one of them appears in the table
and the given avoiding word avoids both formulas.
Some of these formulas are avoided by $b_3$ and the proof uses Cassaigne's algorithm~\cite{cassaignealgo} as in Section~\ref{sec:b3}.
The other formulas are each avoided by the image by a uniform morphism
of either any infinite $\paren{\frac54^+}$-free word $w_5$ over $\Sigma_5$
or any infinite $\paren{\frac75^+}$-free word $w_4$ over $\Sigma_4$. We refer to~\cite{Ochem2004,O16}
for details about the technique to prove avoidance with morphic images of Dejean words.\\

\begin{center}
\begin{tabular}{|l|l|l|l|}
\hline
Formula & Closed under & Avoidability & Avoiding\\
 & reversal? & exponent & word\\ \hline
$ABA.BCB.CAC$ & yes & 1.5 & $b_3$\\ \hline
$ABCA.BCAB.CBAC$ & no & 1.333333333 & $b_3$\\\hline
$ABCA.BAB.CAC$ & yes & 1.414213562 & $g_v(w_4)$\\ \hline
$ABCA.BAB.CBC$ & no & 1.430159709 & $g_w(w_4)$\\ \hline
$ABCA.BAB.CBAC$ & no & 1.381966011 & $g_x(w_5)$\\ \hline
$ABCBA.CABC$ & no & 1.361103081 & $g_y(w_5)$\\ \hline
$ABCBA.CAC$ & yes & 1.396608253 & $g_z(w_5)$\\ \hline
\end{tabular}
\end{center}




\noindent
\begin{minipage}[b]{0.21\linewidth}
\centering
$$\begin{array}{c}
 g_v\\
 0\to\texttt{01220},\\
 1\to\texttt{01110},\\
 2\to\texttt{00212},\\ 
 3\to\texttt{00112}.\\ 
\end{array}$$
\end{minipage}
\begin{minipage}[b]{0.2\linewidth}
\centering
$$\begin{array}{c}
 g_w\\
 0\to\texttt{02111},\\
 1\to\texttt{01121},\\
 2\to\texttt{00222},\\ 
 3\to\texttt{00122}.\\ 
\end{array}$$
\end{minipage}
\begin{minipage}[b]{0.2\linewidth}
\centering
$$\begin{array}{c}
 g_x\\
 0\to\texttt{021110},\\
 1\to\texttt{012221},\\
 2\to\texttt{011120},\\ 
 3\to\texttt{002211},\\
 4\to\texttt{001122}.\\ 
\end{array}$$
\end{minipage}
\begin{minipage}[b]{0.18\linewidth}
\centering
$$\begin{array}{c}
 g_y\\
 0\to\texttt{022},\\
 1\to\texttt{021},\\
 2\to\texttt{012},\\ 
 3\to\texttt{011},\\
 4\to\texttt{000}.\\ 
\end{array}$$
\end{minipage}
\begin{minipage}[b]{0.2\linewidth}
\centering
$$\begin{array}{c}
 g_z\\
 0\to\texttt{120201},\\
 1\to\texttt{100002},\\
 2\to\texttt{022221},\\ 
 3\to\texttt{011112},\\
 4\to\texttt{001122}.\\ 
\end{array}$$
\end{minipage}

\section{A counter-example to a conjecture of Grytczuk}\label{sec:disprove}

Grytczuk~\cite{G07} considered the notion of pattern avoidance on graphs. This generalizes the definition
of nonrepetitive coloring, which corresponds to the pattern $AA$.
Given a pattern $p$ and a graph $G$, the avoidability index $\lambda(p,G)$ is the smallest number of colors
needed to color the vertices of $G$ such that every non-intersecting path in $G$ induces a word avoiding $p$.

We think that the natural framework is that of directed graphs, and we consider only non-intersecting paths that are oriented
from a starting vertex to an ending vertex. This way, $\lambda(p)=\lambda\paren{p,\overrightarrow{P}}$ where $\overrightarrow{P}$
is the infinite oriented path with vertices $v_i$ and arcs $\overrightarrow{v_iv_{i+1}}$, for every $i\ge0$.
The directed graphs that we consider have no loops and no multiple arcs, since they do not modify the set of non-intersecting oriented paths.
However, opposite arcs (i.e., digons) are allowed. Thus, an undirected graph is viewed as a symmetric directed graph: for every pair of distinct vertices
$u$ and $v$, either there exists no arc between $u$ and $v$, or there exist both the arcs $\overrightarrow{uv}$ and $\overrightarrow{vu}$.
Let $P$ denote the infinite undirected path. We are nitpicking about directed graphs because, even though
$\lambda\paren{AA,\overrightarrow{P}}=\lambda(AA,P)=3$, there exist patterns such that 
$\lambda\paren{p,\overrightarrow{P}}<\lambda(p,P)$. For example, $\lambda(ABCACB)=\lambda\paren{ABCACB,\overrightarrow{P}}=2$ and $\lambda(ABCACB,P)=3$.

We do not attempt the hazardous task of defining a notion of avoidance for formulas on graphs.

A conjecture of Grytczuk~\cite{G07} says that for every avoidable pattern $p$, there exists a function $g$ such that $\lambda(p,G)\le g(\Delta(G))$,
where $G$ is an undirected graph and $\Delta(G)$ denotes its maximum degree. Grytczuk~\cite{G07} obtained that his conjecture holds for doubled patterns.

As a counterexample, we consider the pattern $ABACADABCA$ which is $2$-avoidable by the result in Section~\ref{sec:2}.
Of course, $ABACADABCA$ is not doubled because of the isolated variable $D$.
Let us show that $ABACADABCA$ is unavoidable on the infinite oriented graph $\overrightarrow{G}$ with vertices $v_i$ and arcs $\overrightarrow{v_iv_{i+1}}$
and $\overrightarrow{v_{100i}v_{100i+2}}$, for every $i\ge0$. Notice that $\overrightarrow{G}$ is obtained from $\overrightarrow{P}$
by adding the arcs $\overrightarrow{v_{100i}v_{100i+2}}$.
The constant $100$ in the construction is arbitrary and can be replaced by any constant.

Suppose that $\overrightarrow{G}$ is colored with $k$ colors.
Consider the factors in the subgraph $\overrightarrow{P}$ induced by the paths from $v_{300ik+1}$ to $v_{300ik+200k+1}$, for every $i\ge0$.
Since these factors have bounded length, the same factor appears on two disjoint such paths $p_l$ and $p_r$ (such that $p_l$ is on the left of $p_r$).
Notice that $p_l$ contains $2k+1$ vertices with index $\equiv 1\pmod{100}$.
By the pigeon-hole principle, $p_l$ contains three such vertices with the same color $a$.
Thus, $p_l$ contains an occurrence of $ABACA$ such that $A\mapsto a$ on vertices with index $\equiv 1\pmod{100}$.
The same is true for $p_r$. In $\overrightarrow{G}$, the occurrences of $ABACA$ in $p_l$ and $p_r$ imply an occurrence of $ABACADABCA$
since we can skip an occurrence of the variable $A$ in $p_l$ thanks to some arc of the form $\overrightarrow{v_{100j}v_{100j+2}}$.

This shows that $ABACADABCA$ is unavoidable on $\overrightarrow{G}$. So Grytczuk's conjecture is disproved since $\overrightarrow{G}$ has maximum degree $3$.
It is also a counterexample to Conjecture 6 in~\cite{grasshopper} which states that every avoidable pattern is avoidable on the infinite graph
with vertices $\acc{v_0,v_1,\ldots}$ and the arcs $\overrightarrow{v_iv_{i+1}}$ and $\overrightarrow{v_iv_{i+2}}$ for every $i\ge0$. 

\section{A palindrome with index $4$}\label{sec:pal}

Mikhailova~\cite{M13} considered the largest avoidability index $\cal{P}$ of an avoidable pattern that is a palindrome.
She proved that $\mathcal{P}\le 16$.
An obvious lower bound is $\mathcal{P}\ge\lambda(AA)=3$.
For a better lower bound, we consider the palindromic pattern $ABCADACBA$ or, equivalently,
the ternary formula $f=ABCA.ACBA$.
Since it is a ternary formula, $f$ is $4$-avoidable.
It remains to show that $f$ is not $3$-avoidable.
Let $w$ be a ternary recurrent word avoiding $f$.
Suppose to get a contradiction that $w$ contains a square $uu$. 
Then there exists a non-empty word $v$ such that $uuvuu$ is a factor of $w$.
Thus, $w$ contains an occurrence of $f$ given by the morphism $A\mapsto u,B\mapsto u,C\mapsto v$.
This contradiction shows that $w$ is square-free.
A computer check shows that no infinite ternary square-free word avoids $f$.
This holds even if we forbid only squares and every occurrence $h$ of $f$ such that
$|h(A)|=1$ and $|h(B)|+|h(C)|\le5$.
Thus, $\mathcal{P}\ge\lambda(ABCADACBA)=\lambda(ABCA.ACBA)=4$.

\bibliographystyle{plain} 
\bibliography{biblio}
%
%
%
%
%
\end{document}